\documentclass[letterpaper,11pt]{article}
\usepackage[utf8]{inputenc}
\usepackage{amsmath, amsthm, amssymb, thm-restate}
\usepackage{algorithmicx}
\usepackage[table,xcdraw]{xcolor}
\usepackage[vlined,linesnumbered]{algorithm2e}
\usepackage{setspace}
\usepackage{mathtools}
\usepackage[utf8]{inputenc}
\usepackage{thm-restate}
\usepackage[numbers]{natbib}
\usepackage{comment} 
\usepackage[most]{tcolorbox} 
\usepackage{xfrac}
\usepackage{hyperref}
\usepackage{multirow}
\usepackage{caption}
\usepackage{bm}
\usepackage{newfloat}
\usepackage{enumitem}
\usepackage{dblfloatfix} 
\usepackage{wrapfig}

\usepackage{fancyhdr}

\SetCommentSty{mycommfont}
\let\oldnl\nl
\newcommand{\nonl}{\renewcommand{\nl}{\let\nl\oldnl}}

\usepackage[margin=1in]{geometry}

\allowdisplaybreaks

\definecolor{mygreen}{RGB}{10,110,230}
\definecolor{myred}{RGB}{10,110,230}

\hypersetup{
     colorlinks=true,
     citecolor= mygreen,
     linkcolor= myred
}

\renewcommand{\epsilon}{\varepsilon}

\DeclareMathOperator{\E}{\ensuremath{\normalfont \textbf{E}}}

\newcommand{\hiddencomment}[1]{}

\newcommand{\mc}[1]{\ensuremath{\mathcal{#1}}}

\newcommand{\ORS}[0]{\ensuremath{\mathsf{ORS}}}
\newcommand{\RS}[0]{\ensuremath{\mathsf{RS}}}
\newcommand{\RuzsaSzemeredi}[0]{Ruzsa-Szemer\'edi}

\newcommand{\Gadd}[0]{\ensuremath{G_{\text{\normalfont add}}}}
\newcommand{\Gdel}[0]{\ensuremath{G_{\text{\normalfont del}}}}
\newcommand{\Hcert}[0]{\ensuremath{H_{\text{\normalfont cert}}}}
\newcommand{\cupdates}[0]{\ensuremath{c_{\text{\normalfont updates}}}}
\DeclareMathOperator*{\argmin}{arg\,min}

\DeclareMathOperator{\poly}{poly}

\usepackage[noabbrev,nameinlink]{cleveref}
\crefname{lemma}{Lemma}{Lemmas}
\crefname{theorem}{Theorem}{Theorems}
\crefname{property}{Property}{Properties}
\crefname{claim}{Claim}{Claims}
\crefname{definition}{Definition}{Definitions}
\crefname{observation}{Observation}{Observations}
\crefname{proposition}{Proposition}{Propositions}
\crefname{assumption}{Assumption}{Assumptions}
\crefname{line}{Line}{Lines}
\crefname{figure}{Figure}{Figures}
\crefname{conj}{Conjecture}{Conjectures}
\crefname{problem}{Problem}{Problems}
\creflabelformat{property}{(#1)#2#3}
\crefname{equation}{}{}
\crefname{section}{Section}{Sections}
\crefname{appendix}{Appendix}{Appendices}
\crefname{algCounter}{Algorithm}{Algorithms}
\Crefname{algCounter}{Algorithm}{Algorithms}

\newtheorem{theorem}{Theorem}
\newtheorem{lemma}{Lemma}
\newtheorem{proposition}[lemma]{Proposition}

\newtheorem{definition}[lemma]{Definition}
\newtheorem{claim}[lemma]{Claim}

\newtheorem{observation}[lemma]{Observation}

\newtheorem{open}{Open Problem}

\definecolor{mylightgray}{RGB}{240,240,240}

\algnewcommand{\IIf}[2]{\textbf{if} #1 \textbf{then} #2}
\algnewcommand{\EndIIf}{\unskip\ \algorithmicend\ \algorithmicif}
\algnewcommand{\IElse}[1]{\textbf{else} #1}

\newenvironment{graytbox}{
\par\addvspace{0.1cm}
\begin{tcolorbox}[width=\textwidth,
                  boxsep=5pt,
                  left=1pt,
                  right=1pt,
                  top=2pt,
                  bottom=2pt,
                  boxrule=1pt,
                  arc=0pt,
                  colback=mylightgray,
                  colframe=black,
                  ]
}{
\end{tcolorbox}
}

\newenvironment{whitetbox}{
\par\addvspace{0.1cm}
\begin{tcolorbox}[width=\textwidth,
                  boxsep=5pt,
                  left=1pt,
                  right=1pt,
                  top=2pt,
                  bottom=2pt,
                  boxrule=1pt,
                  arc=0pt,
                  colframe=black,
                  colback=white
                  ]
}{
\end{tcolorbox}
}

\newcounter{algCounter}

\newenvironment{algenv}[2]{
    \begin{whitetbox}
        \refstepcounter{algCounter}
        \textbf{Algorithm~\thealgCounter:} #1
        \label{#2}
        
        \vspace{-0.2cm}
        \noindent\rule{\linewidth}{1pt}
        \begin{algorithm}[H]
}{
        \end{algorithm}
    \end{whitetbox}
}

\makeatletter
\renewcommand{\paragraph}{%
  \@startsection{paragraph}{4}%
  {\z@}{10pt}{-1em}%
  {\normalfont\normalsize\bfseries}%
}
\makeatother

\makeatletter
\patchcmd{\@algocf@start}
  {-1.5em}
  {0pt}
  {}{}
\makeatother

\title{Fully Dynamic Matching and Ordered Ruzsa-Szemer\'edi Graphs}

\author{
Soheil Behnezhad\\{\em Northeastern University} \and 
Alma Ghafari\\{\em Northeastern University} \and
}

\date{}

\begin{document}

\maketitle

\thispagestyle{empty}

\begin{abstract}
{\setlength{\parskip}{0.2cm}

We study the fully dynamic maximum matching problem. In this problem, the goal is to efficiently maintain an approximate maximum matching of a graph that is subject to edge insertions and deletions. Our focus is particularly on algorithms that maintain the edges of a $(1-\epsilon)$-approximate maximum matching for an arbitrarily small constant $\epsilon > 0$. Until recently, the fastest known algorithm for this problem required $\Theta(n)$ time per update where $n$ is the number of vertices. This bound was slightly improved to $n/(\log^* n)^{\Omega(1)}$ by Assadi, Behnezhad, Khanna, and Li [STOC'23] and very recently to $n/2^{\Omega(\sqrt{\log n})}$ by Liu [FOCS'24]. Whether this can be improved to $n^{1-\Omega(1)}$ remains a major open problem.

In this paper, we introduce {\em Ordered Ruzsa-Szemerédi (ORS)} graphs (a generalization of Ruzsa-Szemerédi graphs) and show that the complexity of dynamic matching is closely tied to them. For $\delta > 0$, define $\ORS(\delta n)$ to be the maximum number of matchings $M_1, \ldots, M_t$, each of size $\delta n$, that one can pack in an $n$-vertex graph such that each matching $M_i$ is an {\em induced matching} in subgraph $M_1 \cup \ldots \cup M_{i}$. We show that there is a randomized algorithm that maintains a $(1-\epsilon)$-approximate maximum matching of a fully dynamic graph in
$$
    \widetilde{O}\left( \sqrt{n^{1+\epsilon} \cdot \ORS(\Theta_\epsilon(n))} \right)
$$
amortized update-time. 

While the value of $\ORS(\Theta(n))$ remains unknown and is only upper bounded by $n^{1-o(1)}$, the densest construction known from more than two decades ago only achieves $\ORS(\Theta(n)) \geq n^{1/\Theta(\log \log n)} = n^{o(1)}$ [Fischer et al. STOC'02]. If this is close to the right bound, then our algorithm achieves an update-time of $\sqrt{n^{1+O(\epsilon)}}$, resolving the aforementioned longstanding open problem in dynamic algorithms in a strong sense.

}
\end{abstract}

{
\clearpage
\hypersetup{hidelinks}
\vspace{1cm}
\renewcommand{\baselinestretch}{0.1}
\setcounter{tocdepth}{2}
\tableofcontents{}
\thispagestyle{empty}
\clearpage
}

\setcounter{page}{1}

\clearpage

\section{Introduction}

We study {\em dynamic} algorithms for the {\em maximum matching} problem, a cornerstone of combinatorial optimization. Given a graph $G=(V, E)$ a matching $M \subseteq E$ is a collection of vertex-disjoint edges. A maximum matching is a matching of largest possible size in $G$. We study the maximum matching problem in fully dynamic graphs. In this problem, the input graph $G$ changes over time via a sequence of edge insertions and deletions. The goal is to maintain an (approximate) maximum matching of $G$ at all times, without spending too much time after each update.

\paragraph{Background on Dynamic Matching:} The dynamic matching problem has received a lot of attention over the last two decades \cite{OnakR10,BaswanaGS11,BaswanaGS18,NeimanS-STOC13,GuptaPeng-FOCS13,BhattacharyaHI-SiamJC18,BernsteinS-ICALP15,BernsteinSteinSODA16,BhattacharyaHN-SODA17,BhattacharyaHN-STOC16,Solomon-FOCS16,CharikarS18-ICALP,ArarCCSW-ICALP18,BernsteinFH-SODA19,BehnezhadDHSS-FOCS19,BehnezhadLM-SODA20,Wajc-STOC20,BernsteinDL-STOC21,BhattacharyaK21-ICALP21,RoghaniSW-ITCS22,Kiss-ITCS22,GrandoniSSU-SOSA22,BehnezhadK-SODA22,BhattacharyaKSW-SODA23,Behnezhad-SODA23,AzarmehrBR-SODA24,BhattacharyaKS-FOCS23}. There is a relatively simple algorithm that maintains a $(1-\epsilon)$-approximate maximum matching, for any fixed $\epsilon > 0$, in just $O(n)$ time per update (see e.g. \cite{GuptaPeng-FOCS13}). The update-time can be significantly improved if we worsen the approximation. For instance, a 1/2-approximation can be maintained in $\poly\log n$ time per update \cite{Solomon16,BaswanaGS11,BehnezhadDHSS-FOCS19}, or an (almost) 2/3-approximation can be maintained in $O(\sqrt{n})$ time \cite{BernsteinSteinSODA16}. Nonetheless, when it comes to algorithms with approximation ratio better than $2/3$, the update-time stays close to $n$. A slightly sublinear algorithm was proposed by \citet*{AssadiBKL23} which runs in $n/(\log^* n)^{\Omega(1)}$ time per update and maintains a $(1-o(1))$-approximation. In a very recent paper, \citet*{Yang24} improved this to $n/2^{\Omega(\sqrt{\log n})}$ via a nice connection to  algorithms for the online matrix-vector multiplication (OMv) problem. Despite this progress, the following remains a major open problem:

\begin{open}\label{open:dynamic1+eps}
    Is it possible to maintain a $(1-\epsilon)$-approximate maximum matching in a fully dynamic graph, for any fixed $\epsilon > 0$, in $n^{1-\Omega(1)}$ update-time?
\end{open}

We note that there is an orthogonal line of work on fully dynamic algorithms that instead of maintaining the edges of the matching, maintains only its size \cite{Behnezhad-SODA23,BhattacharyaKSW-SODA23,BhattacharyaKS-FOCS23}. For this easier version of the problem, \citet*{BhattacharyaKS-FOCS23} positively resolved the open problem above. However, their algorithm crucially relies on only estimating the size and does not work for the problem of maintaining the edges of the matching. We refer interested readers to \cite{Behnezhad-SODA23} where the difference between the two versions of the problem is mentioned.

\subsection{Our Contributions}

\paragraph{Contribution 1: Dynamic Matching.} In this paper, we make progress towards \Cref{open:dynamic1+eps} by presenting a new algorithm whose update time depends on the density of a certain class of graphs that we call Ordered \RuzsaSzemeredi{} (ORS) graphs, a generalization of the well-known \RuzsaSzemeredi{} (RS) graphs.

Let us start by defining RS graphs.

\begin{definition}[Ruzsa-Szemer\'edi Graphs]\label{def:RS}
    An $n$-vertex graph $G=(V, E)$ is an $\RS_n(r, t)$ graph if its edge-set $E$ can be decomposed into $t$ edge-disjoint induced matchings each of size $r$. We use $\RS_n(r)$ to denote the maximum $t$ for which $\RS_n(r, t)$ graphs exists.
\end{definition}

Instead of each matching being an induced matching in the whole graph, the edges of an ORS graph should be decomposed into an ordered list of matchings such that each matching is induced only with respect to the previous matchings in the ordering. The following formalizes this.

\begin{definition}[Ordered Ruzsa-Szemer\'edi Graphs]\label{def:ORS}
    An $n$-vertex graph $G=(V, E)$ is an $\ORS_n(r, t)$ graph if its edge-set $E$ can be decomposed into an ordered list of $t$ edge-disjoint matchings $M_1, \ldots, M_t$ each of size $r$ such that for every $i \in [t]$, matching $M_i$ is an induced matching in $M_1 \cup \ldots \cup M_i$. We use $\ORS_n(r)$ to denote the maximum $t$ for which $\ORS_n(r, t)$ graphs exists.
\end{definition}

Note that every $\RS_n(r, t)$ graph is an $\ORS_n(r, t)$ graph but the reverse is not necessarily true.

Our main result can now be stated as follows:

\begin{graytbox}

\begin{restatable}{thm}{fullydynamic}
    
\begin{theorem}\label{thm:fully-dynamic}
    Let $\epsilon > 0$ be fixed. There is a fully dynamic algorithm that maintains the edges of a $(1-\epsilon)$-approximate maximum matching in $O\left(\sqrt{n^{1+\epsilon} \cdot \ORS_n(\Theta_\epsilon(n))} \poly(\log n) \right)$ amortized update-time. The algorithm is randomized but works against adaptive adversaries.
\end{theorem}
\end{restatable}
\end{graytbox}

 To understand the update-time in \cref{thm:fully-dynamic}, we need to understand the density of ORS graphs for linear size matchings. Let $0 < c < 1/5$ be a constant. Since $\ORS_n(cn) \geq \RS_n(cn)$ as every RS graph is also an ORS graph with the same parameters, it is natural to first look into the more well-studied case of RS graphs. In other words, how dense can RS graphs with linear size matchings be? We note that this question has been of interest to various communities from property testing \cite{FischerLNRRS02} to  streaming algorithms \cite{GoelKK12,AssadiBKL23,Assadi22} to additive combinatorics \cite{fox2015graphs}. Despite this, the value of $\RS_n(cn)$ remains widely unknown. 
 
 The best lower bound on $\RS_n(cn)$--- i.e., the densest known construction---is that of \cite{FischerLNRRS02} from more than two decades ago which shows $\RS_n(cn) \geq n^{\Omega_c(1/\log \log n)} = n^{o(1)}$. This is indeed the densest known construction of ORS graphs we are aware of too. If this turns out to be the right bound, \cref{thm:fully-dynamic} implies a $(1-\epsilon)$-approximation in $\widetilde{O}(n^{1/2+\epsilon})$ time, an almost quadratic improvement over prior near-linear in $n$ algorithms of \cite{AssadiBKL23,Yang24}. In fact, we note that so long as $\ORS_n(cn)$ is moderately smaller than $n$ (say $\ORS_n(\Theta_\epsilon(n)) \ll n^{1-\Omega(1)}$) \cref{thm:fully-dynamic} still implies a truly sublinear in $n$ update-time algorithm, positively resolving \Cref{open:dynamic1+eps}.

Finally, we note that there is a long body of work on proving conditional lower bounds for dynamic problems. For instance, the OMv conjecture can be used to prove that  maintaining an exact maximum matching requires near-linear in $n$ update-time \cite{HenzingerKNS15}. Adapting these lower bounds to the $(1-\epsilon)$-approximate maximum matching problem has remained open since then. Our \cref{thm:fully-dynamic} implies that proving such lower bounds either requires a strong lower bound of near-linear on $\ORS_n(\Theta(n))$, or requires a conjecture that implies this. 

\vspace{-0.2cm}

\paragraph{Contribution 2: Upper Bounding ORS.} Unfortunately there is a huge gap between existing lower and upper bounds for $\RS_n(cn)$ (and as a result also for $\ORS_n(cn)$). The best known upper bound on $\RS_n(c n)$ for linear size matchings follows from the improved triangle-removal lemma of \citet{fox2011new} which implies  $\RS_n(cn) \leq n/\log^{(\ell)} n$ for $\ell = O(\log (1/c))$ where $\log^{(x)}$ is the iterated log function. We note that this result is implicit in \cite{fox2011new} and was mentioned in the paper of \cite{fox2015graphs}. To our knowledge, this upper bound does not carry over to ORS graphs (we briefly discuss this at the beginning of \cref{sec:ub}). Our second result is a similar upper bound for ORS albeit with a worse dependence on constant $c$.

\begin{graytbox}
\begin{theorem}\label{thm:ORS-ub}
    For any $c > 0$, it holds that $\ORS_n(cn) = O(n / \log^{(\ell)} n)$ for some $\ell = \poly(1/c)$.
\end{theorem}
\end{graytbox}

Since every RS graph is also an ORS graph with the same parameters, \cref{thm:ORS-ub} immediately implies the same upper bound for $\RS_n(cn)$. Note that this implication is not a new result, but the proof is different from that of \cite{fox2011new} and is closer to the arguments in \cite{fox2015graphs}.

\subsection{Perspective: ORS vs RS}

Summarizing the above-mentioned bounds, we have
\begin{equation*}
n^{\Omega_c(1/\log \log n)} \stackrel{\text{\cite{FischerLNRRS02}}}{\leq} \RS_n(cn) \leq \ORS_n(cn) \stackrel{\text{\cref{thm:ORS-ub}}}{\leq} O(n / \log^{(\poly(1/c))} n).
\end{equation*}

While the value of $\RS_n(cn)$ remains widely unknown, one might argue that $\RS_n(cn) = n^{o(1)}$ is a plausible outcome, given that the construction of \cite{FischerLNRRS02} has resisted any improvements for over two decades despite significant interest. But should we believe that $\ORS_n(cn)$ is also small in this case?  Unfortunately the authors could not prove any formal relation between the densities of RS and ORS graphs beyond the upper bound of \cref{thm:ORS-ub}. In particular, we believe the following is an extremely interesting question for future work:

\begin{open}\label{open:ORSpoly}
    Is it true that for fixed $\epsilon > 0$, $$\ORS_n(\epsilon n) \leq \poly(\RS_n(\Theta_\epsilon(n)))?$$
\end{open}

In the event that the answer to \Cref{open:ORSpoly} is positive and $\RS_n(cn) = n^{o(1)}$ for any fixed $c > 0$,  we also get that $\ORS_n(cn) = n^{o(1)}$. Therefore \cref{thm:fully-dynamic} would imply an $n^{1/2+O(\epsilon)}$ time algorithm in this case.

In the event that the answer to \Cref{open:ORSpoly} is negative, one might wonder whether we can improve \cref{thm:fully-dynamic} by parametrizing it based on RS instead of ORS. Put differently, suppose that $\ORS_n(cn) = n^{1-o(1)}$ and $\RS_n(cn) = n^{o(1)}$. Can we somehow utilize the sparsity of RS graphs (instead of ORS graphs) in this case to improve existing dynamic matching algorithms? We start \cref{sec:techniques} by providing an input construction which informally shows ORS is the right parameter for \cref{thm:fully-dynamic} even if RS graphs turn out to be much sparser.

\subsection{Connections and Recent Developments}

\paragraph{Connection to Sublinear Time Algorithms:} Our algorithm crucially relies on sublinear time algorithms; particularly the algorithm of the first author in \cite{behnezhad2021} for estimating the size of maximum matching. While sublinear-time algorithms have also played a crucial role in the dynamic matching algorithms of \cite{Behnezhad-SODA23,BhattacharyaKSW-SODA23,BhattacharyaKS-FOCS23}, these algorithms maintain only the size of the matching as opposed to its edges. Our algorithm, on the other hand, maintains the matching explicitly yet relies on matching size estimators crucially. On a high level, our algorithm can tolerate the time needed to {\em find} a large matching in a specific induced subgraph so long as a large matching is guaranteed to exist there. We use the fast sublinear-time matching algorithm of \cite{behnezhad2021} to first verify existence of a large matching, then spend the time needed if this matching is large.

\paragraph{Connection to a Lower Bound of \citet*{Yang24}:} In his very recent work, \citet*{Yang24} showed that under the approximate OMv conjecture, there is no algorithm that maintains a $(1-\epsilon)$-approximate maximum matching in $n^{1-\Omega(1)} \poly(1/\epsilon)$ update-time. This might, at the first glance, appear to contradict \cref{thm:fully-dynamic} which runs in $n^{1-\Omega(1)}$ time if it so happens that $\ORS(\Theta(n))=n^{1-\Omega(1)}$. However, we note that there is no contradiction even if approximate OMv conjecture turns out to be true and $\ORS(\Theta(n))=n^{1-\Omega(1)}$ at the same time! The reason is that when parameter $\epsilon$ is sufficiently sub-constant, dense constructions of RS (and consequently ORS) with $\Theta(n^2)$ edges already exist \cite{AlonMS12}. So our \cref{thm:fully-dynamic} is only potentially useful for fixed (or mildly sub-constant) $\epsilon$. On the other hand, the reduction of \cite{Yang24} sets $\epsilon = n^{-\Omega(1)}$, thus only targets the sub-constant regime.

Interestingly, while the result of \cite{Yang24} removes the combinatorial structure of the dynamic matching problem and reduces it to a purely algebraic question for the sub-constant regime of $\epsilon$, our work shows that the fixed regime of $\epsilon$ is, in fact, a completely combinatorial question.

\paragraph{Recent Developments:} After the first version of this paper appeared online, the nice follow-up works of \citet*{AssadiKhanna-ORS-Arxiv} and \citet*{Kiss-ORS} improved our update-time from $O(\sqrt{n^{1+o(1)}\ORS(\Theta(n)})$ to $O(n^{o(1)}\ORS(\Theta(n)))$. Importantly, their algorithms run in $n^{o(1)}$ update-time under $\ORS(\Theta(n))=n^{o(1)}$, showing that ORS graphs fully characterize the hardness of the dynamic matching problem.

\section{Our Techniques}\label{sec:techniques}

In this section, we provide an informal overview of our algorithm for \cref{thm:fully-dynamic} as well as the upper bound of \cref{thm:ORS-ub}.

Before describing the intuition behind our algorithm of \cref{thm:fully-dynamic}, let us start with a sequence of updates that, in a sense, explains why existence of dense ORS graphs would make it challenging to maintain a $(1-\epsilon)$-approximate maximum matching in a fully dynamic setting.

\paragraph{Why ORS graphs are seemingly hard:} Consider a fully dynamic input graph $G=(V, E)$ that is composed of two types of vertices: the {\em ORS vertices} $V_{ORS} \subseteq V$ which is a subset of $n$ vertices, and the {\em singleton vertices} $V_S$ which is a subset of $(1-2\epsilon)n$ vertices. We start by inserting an ORS graph in the induced subgraph $G[V_{ORS}]$. Namely, take an $\ORS_n(\epsilon n, t)$ graph on $n$ vertices. We make the induced subgraph $G[V_{ORS}]$ isomorphic to this ORS graph by inserting its edges one by one to $G[V_{ORS}]$. Let $M_1, \ldots, M_t$ be the ordered induced matchings of $G[V_{ORS}]$ as defined in \cref{def:ORS}. Then the sequence of updates is as follows:
\begin{itemize}
    \item For $i = t$ to $1$:
\begin{itemize}
    \item Delete all existing edges of $V_S$.
    \item If $i \not= t$, delete the edges of $M_{i+1}$.
    \item Let $M_i$ be the {\em current} induced matching in $G[V_{ORS}]$.
    \item Insert a perfect matching between the $(1 - 2\epsilon) n$ vertices of $V_{ORS}$ left unmatched by $M_i$ and the $(1-2\epsilon)n$ vertices in $V_{S}$.
\end{itemize}
\end{itemize}

Take the graph after the iteration $i$ of the for loop. Note that there is a perfect matching in $G$: match all singleton vertices to the $V_{ORS}$ vertices not matched by $M_i$, and match the rest of the vertices in $V_{ORS}$ through $M_i$. Importantly, since all matchings $M_{i+1}, \ldots, M_t$ have already been deleted from the graph, $M_i$ must be an induced matching of the remaining graph $G$ (the other matchings $M_1, \ldots, M_{i-1}$ cannot have any edge with both endpoints matched by $M_i$ due to \cref{def:ORS}). Because of this, it can be confirmed that any $(1-\epsilon/2)$-approximate maximum matching of $G$ must include at least half of the edges of $M_i$. The naive algorithm for finding an edge of $M_i$ for some vertex $v$ would scan the neighbors of $v$, which could take $\Omega(t)$ time per vertex (as this is the degrees in the ORS graph) and thus $n t$ time in total after every iteration of the loop consisting of $n$ updates. Hence, the amortized update-time of this algorithm must be at least $\Omega(nt/n) = \Omega(\ORS_n(\epsilon n))$.

The input construction above implies that to maintain a $(1-\epsilon)$-approximation of maximum matching, either we have to find a way to identify induced matchings of an ORS graph fast (without scanning the neighbors of each vertex) which appears extremely challenging, or we have to parameterize our algorithm's update-time by $t$, the density of ORS graphs. We take the latter approach in this work.

\paragraph{Overview of our algorithm for \cref{thm:fully-dynamic}:} Let us for this informal overview of our algorithm assume that the maximum matching size is at least $\Omega(n)$. Having this assumption (which comes w.l.o.g. as stated by \cref{prop:multiplicative}) allows us to find the matching once, do nothing for the next $\epsilon n$ updates, and then repeat without hurting the size of the approximate matching that we find by more than a $1+O(\epsilon)$ factor.

As it is standard by now, to find a $(1-\epsilon)$-approximate matching it suffices to design an algorithm that given a subset $U \subseteq V$, finds a constant approximate maximum matching in $G[U]$. If this algorithm runs in $T$ time, we can find a $(1-\epsilon)$-approximate maximum matching of the whole graph also in $O_\epsilon(T)$ time. Amortized over $\epsilon n$ updates, this runs in $O(T/n)$ total time for constant $\epsilon > 0$. However, just like the challenging example discussed above, in case the maximum matching in the induced subgraph $G[U]$ is an induced matching, we do not know how to find a constant fraction of its edges without spending $\Omega(n^2)$ time. However, if we manage to bound the total number of such hard subsets $U$ for which we spend a lot of time, then we can bound the update-time of our algorithm. Intuitively, we would like to guarantee that if our algorithm takes $\Omega(n^2)$ time to solve an instance $G[U]$, then the maximum matching in $G[U]$ must be an induced matching of the graph $G$, and charge these heavy computations to ORS which provides an upper bound on the number of edge-disjoint such induced matchings. However, there are two main problems: (1) it may be that the maximum matching in $G[U]$ is not an induced matching of the graph, yet it is sparse enough that it is hard to find; (2) even if $G[U]$ forms an induced matching, we have to ensure that its edges do not belong to previous induced matchings that we have charged, as ORS only bounds the number of {\em edge-disjoint} ordered induced matchings.

For the first problem discussed above, we present an algorithm that runs in (essentially) $O(n^2/d)$ time to find the maximum matching in $G[U]$. Here $d$ is a parameter that depends on the structure of $G[U]$ that measures how easy it is to find an approximate maximum matching of $G[U]$. Intuitively, if the average degree within $G[U]$ is $d$, we can random sample pairs of vertices to add to the matching. If each vertex is adjacent to $d$ others, we only need $O(n/d)$ samples to match it and the algorithm runs in $O(n^2/d)$ time. Of course, this can take up to $\Omega(n^2)$ time if $G[U]$ is sparse -- e.g. when it is an induced matching. Then instead of charging a matching in $G[U]$ that is an induced matching, we charge this matching of average degree at most $d$ inside, which we call a {\em certificate}. Let $M_1, \ldots, M_t$ be the certificates that we charge and let $d_i$ be the average degree of the $i$-th matching and suppose that these matchings are edge-disjoint. We show in our update-time analysis that $\sum_{i=1}^t 1/d_i$ can be at most $\ORS(\Theta(n))$, and therefore the total time spent by our algorithm during a phase can be upper bounded by $n^2 \ORS(\Theta(n))$.

For the second problem, or in other words, to ensure that the certificate matchings that we charge are edge-disjoint, we maintain a set $\Hcert$ and add all edges of any matching that we charge to this set. Thereafter, before solving $G[U]$, we first go over the edges stored in this certificate set and see whether they can be used to find a large matching in $G[U]$. If they do, we do not run the random sampling algorithm discussed above. If not, the matching that we find must be edge-disjoint. We have to be careful that we do not make $\Hcert$ too dense though as we spend linear time in the size of $\Hcert$. Our final algorithm resets $\Hcert$ after a certain number of updates. 

We note that our informal discussion of this section hides many (important) details of the final algorithm that we formalize in \cref{sec:dynamic}.

\paragraph{Overview of our upper bound in \cref{thm:ORS-ub}:}

To obtain our upper bound of \cref{thm:ORS-ub}, we partition the matchings into two subsets, $\mc{M}$ and $\mc{M}'$, based on their order in the sequence. A key insight (a variant of which was used in the RS upper bound of \cite{fox2015graphs}) is the following: take a vertex $v$ and suppose that it is matched by some matching $M$ in $\mc{M}'$. If we remove all neighbors of vertex $v$ in $\mc{M}$, then it can be proved that no vertex of matching $M$ is removed because otherwise there must be an edge from $\mc{M}$ that matches two vertices of $M$, violating the inducedness property. Intuitively, this shows that if vertices have large degrees in $\mc{M}$, we can remove a relatively large number of vertices without hurting the matchings that include this vertex. To derive the upper bound, we carefully select a set of pivots based on the degrees in $\mc{M}$, and remove the neighbors of these pivots. We show that this reduces the number of vertices significantly enough, and keeps the size of a small (but sufficiently many) of the matchings unchanged. If the initial number of matchings is so large, we show that we can iteratively applying this procedure. Because the size of matchings do not change but the number of vertices drops, we get that the process should eventually stop. This implies the upper bound on the number of matchings in the starting graph.

\section{Preliminaries}

A {\em fully dynamic} graph $G=(V, E)$ is a graph defined on a fixed vertex set $V$ that is subject to edge insertions and deletions. We assume that each edge update is issued by an {\em adaptive adversary} that can see our algorithm's previous outputs and can accordingly decide on the next update. We use $\mu(G)$ to denote the maximum matching size in G. We say an algorithm has amortized update-time $U$ if the total time spent after $T$ updates is $U \cdot T$ for some sufficiently large $T = \poly(n)$. 

\paragraph{Tools from prior work:} Here we list some of the tools we use from prior work in our result.

The following proposition, implied by the streaming algorithm of \citet*{McGregor05} (see also \cite{BhattacharyaKS-FOCS23} for its dynamic adaptation), shows that to find a $(1-\epsilon)$-approximate matching, it suffices to solve a certain induced matching problem a constant number of times.

\begin{proposition}[Approximation Boosting Framework \cite{McGregor05,BhattacharyaKS-FOCS23}]\label{prop:boosting-apx}
    Let $G=(V, E)$ be any (possibly non-bipartite) $n$-vertex graph. Suppose that for any parameter $\delta \in (0, 1)$ we have an algorithm $\mc{A}(G, U, \delta)$ that provided any vertex subset $U \subseteq V$ with $\mu(G[U]) \geq \delta \cdot n$, finds a matching of size at least $\poly(\delta) \cdot \mu(G[U])$ in $G[U]$. Then for any $\epsilon \in (0, 1)$, there is an algorithm that finds a matching of size at least $\mu(G) - \epsilon n$ in $G$ by making $t = O_{\epsilon}(1)$ adaptive calls 
    $$
    \mc{A}(G, U_1, \delta_1), \ldots, \mc{A}(G, U_t, \delta_t)
    $$
    to algorithm $\mc{A}$. The value of $\delta_i$ in each of these calls is just a function of $\epsilon$, $\delta_i \leq \epsilon/2$, and preparing the vertex subsets $U_1, \ldots, U_t$ can be done in $\widetilde{O}_\epsilon(n)$ total time.
\end{proposition}

We also use the following sublinear-time algorithm of \citet*{behnezhad2021} for estimating the size of maximum matching. We note that even though we use this algorithm in a crucial, our final dynamic algorithm does not maintain just the size, but rather the edges of the matching, explicitly. 

\begin{proposition}[\cite{behnezhad2021}]\label{prop:sublinear}
    Let $G=(V, E)$ be any (possibly non-bipartite) $n$-vertex graph. For any $\epsilon > 0$, there is an algorithm that makes $\widetilde{O}(n \poly(1/\epsilon))$ adjacency matrix queries to $G$ and provides an estimate $\widetilde{\mu}$ of the size of maximum matching $\mu(G)$ in $G$ such that with probability $1-1/\poly(n)$, it holds that $$0.5 \mu(G) - \epsilon n \leq \widetilde{\mu} \leq \mu(G).$$
\end{proposition}

Finally, we use the following proposition which turns an additive approximation into a multiplicative approximation. The proof is based on a vertex sparsification idea for matchings \cite{AssadiKL17,BehnezhadDH20} which was adapted to the dynamic setting in the work of \citet*{Kiss-ITCS22}.

\begin{proposition}[\cite{Kiss-ITCS22,AssadiKL17}]\label{prop:multiplicative}
    Suppose there is an adaptive algorithm \mc{A}, that for any parameter $\epsilon > 0 $ and a fully dynamic $n$-vertex graph $G = (V,E)$, maintains a matching of size $\mu(G)-\epsilon n$ in $Q(n,\epsilon)$ amortized time. Then there is an algorithm that maintains a multiplicative $(1-\epsilon)$-approximation maximum matching of $G$ in $\poly(\log(n), \epsilon) \cdot Q(n,\epsilon^2)$ amortized time per update.
\end{proposition}

\section{Our Dynamic Algorithm for Approximate Maximum Matching}

In this section, we present our algorithm and prove \cref{thm:fully-dynamic}. We start in \cref{sec:static} with a static {\em potentially sublinear-time} algorithm that is one of the main building blocks of our final algorithm. We then formalize our dynamic algorithm in \cref{sec:dynamic}. In \cref{sec:correctness} we prove correctness of our dynamic algorithm and analyze its running time in \cref{sec:runtime}.

\subsection{A Static Potentially-Sublinear Time Algorithm via Random Sampling} \label{sec:static}

In this section, we present a {\em potentially sublinear time} algorithm that given an $n$ vertex graph $G=(V, E)$, finds a matching $M$ that is by an additive factor of at most $\epsilon n$ smaller than the maximum matching of $G$. The algorithm, in addition, returns a {\em certificate} $M_C$, which is another matching of $G$ that we use to bound the running time of the algorithm. Particularly, denoting by $d$ the average degree in the subgraph induced on the vertices of the certificate (i.e., graph $G[V(M_C)]$), the running time of the algorithm overall is $\widetilde{O}(n^{2+2\epsilon}/d)$. If the certificate $M_C$ is close to an induced matching (i.e., $d$ is small), then the algorithm is not better than the trivial algorithm which reads all the edges in quadratic in $n$ time. On the other hand, if $G[V(M_C)]$ is dense, the algorithm runs in sublinear time. This algorithm will be a crucial component of our final dynamic algorithm.

The following is the main lemma of this section.

\newcommand{\Edense}[0]{\ensuremath{E_{\text{\normalfont dense}}}}
\newcommand{\Esparse}[0]{\ensuremath{E_{\text{\normalfont sparse}}}}

\begin{lemma}\label{lem:nd-boosted}
    Let $G=(V, E = \Edense \cup \Esparse)$ be a given $n$-vertex graph  where we have adjacency matrix access to \Edense{}, adjacency list access to \Esparse{}, and $\Edense$ and $\Esparse$ may share edges. For any parameter $\epsilon \in (0, 1)$, there is an algorithm $\textsc{MatchAndCertify}(\Edense, \Esparse, \epsilon)$ that returns a $(1, \epsilon n)$-approximate maximum matching $M$ of $G$ and a {\em certificate} $M_C$. Let $T$ be the running time of the algorithm. Then exactly one of the following two conditions holds:
    \begin{enumerate}[label=$(C\arabic*)$]
        \item\label{itm:c1} $M_C = \emptyset$ and $T = \widetilde{O}_\epsilon(n+|\Esparse|)$.
        \item \label{itm:c2}$M_C$ is a matching in $\Edense$ of size $\Omega_\epsilon(n)$ where $M_C \cap \Esparse = \emptyset$. Additionally, $T= \widetilde{O}_\epsilon(|\Esparse| + n^{2+\epsilon}/d)$ where  $d := |\Edense \cap V^2(M_C)|/|M_C|$ is the average degree of $\Edense[V(M_C)]$.
    \end{enumerate}
\end{lemma}

Towards proving \cref{lem:nd-boosted}, we first prove the following \cref{lem:nd} which solves a slightly simpler subproblem that finds a large matching in a given induced subgraph of $G$. We will later prove \cref{lem:nd-boosted} by combining \cref{lem:nd} with \cref{prop:boosting-apx}.

\begin{lemma}\label{lem:nd}
    Let $G= (V,E)$ be a given $n$-vertex graph to which we have adjacency matrix access. Let $U \subseteq V$ be a given vertex subset, and let $\delta \in (0,1)$ be a given parameter such that $\mu(G[U]) \geq \delta n$. There is an algorithm that finds a matching $M \subseteq G[U]$ of size at least $\delta^2 n$. The algorithm runs in  $\widetilde{O}(n^{2+2\delta}/d)$ time where $d = |E \cap V^2(M)|/|M|$.
\end{lemma}

\begin{proof}
The algorithm is formalized below as \cref{alg:randomsampling}. 
\begin{figure}[h]
\begin{algenv}{\textsc{RandomSampling}$(G[U],\delta)$}{alg:randomsampling}
    \nonl \textbf{Parameter:} $\delta \in (0, 1]$.

    $M \gets \emptyset$, $b \gets 1, U' \gets U$.

    \For{$i=1$ to $1/2\delta$}{

        $M_i \gets \emptyset$
        
        \For{vertex $v \in U'$}{
        
        $S \gets $ Sample $b$ vertices from $U'$ uniformly at random.

        \If {$\exists u \in S$ such that $(u, v) \in E$ and $u$ is available}
        {
        
        $M_i \gets (v,u) \cup M_i$.

        Remove $u$ and $v$ from $U'$.
        
        }

        $M \gets M \cup M_i$.
        
        $b \gets n^{2\delta} b $
    }
    \If {$|M_i| \geq \delta^2 n$}{
        \Return $M_i$.
    }
    }
\end{algenv}
\end{figure}
Note that if the algorithm returns a matching $M_i$, then it is guaranteed that the matching has our desired size of at least $\delta^2 n$. So it suffices to prove that at some point $M_i$ is as large as we want. To see this, note that once the budget $b$ reaches $n$, we will go through all the edges, indicating that $M$, the union of matching $M_1,M_2,\dots, M_{1/(2\delta)}$, is a maximal matching of $G[U]$. Since the size of a maximal matching is at least half the size of a maximum matching, we have $|M| \geq \mu(G[U])/2 \geq \delta n/2$. This means that at least one of the matchings in $\{M_1, ..., M_{1/(2\delta)}\}$ is of size at least $\frac{\delta^2n/2}{1/(2\delta)} = \delta^2 n$.

Next, we focus on relating the running time of the algorithm to the parameter $d$, i.e., the average degree of the output matching of \cref{alg:randomsampling}. Let us first introduce some notation. Let $U'_i$ be the value of the set $U'$ at the end of iteration $i$. Additionally, for every vertex $v$, let $d^i_v$ be the degree of vertex $v$ in graph $G[U'_i]$. We start by proving the following useful claim.

\begin{claim}\label{cl:huhu-hjb}
    It holds with probability $1-1/n$ for every $v \in U'_i$ and every $i$ that $d^i_v = O(n^{1-2\delta(i-1)} \log n)$.
\end{claim}
\begin{proof}
    Fix $v$ and $i$. We prove that either $v \not\in U'_i$ or $d^i_v = O(n^{1-2\delta(i-1)} \log n)$ with probability at least $1-1/n^3$. A union bound over the choices of $v$ and $i$ completes the proof.
    
    Consider iteration $i$ of the algorithm and suppose $v \in U'_i$.  This means that in this iteration we attempt to match vertex $v$ by sampling $b = n^{2\delta(i-1)}$ vertices in $U' \supseteq U'_{i}$ and checking if any of them is a neighbor of $v$. Let $d''$ be the number of neighbors of $v$ in $U'$ when processing $v$ in iteration $i$. We have
    $$
        \Pr[v \text{ not matched by } M_i] \leq \left(1-\frac{d''}{|U'|}\right)^b \leq \left(1-\frac{d''}{n}\right)^b \leq \exp\left( - \frac{d'' b}{n} \right) = \exp\left( - \frac{d''}{n^{1-2\delta(i-1)}} \right).
    $$
    Now consider two cases. If $d'' \geq 3 n^{1-2\delta(i-1)} \log n$, then $v$ is matched in $M_i$ with probability at least $1-1/n^3$ and thus will be removed from $U'$. Otherwise, since $d^i_v \leq d''$, we have $d^i_v = O(n^{1-2\delta(i-1)} \log n)$. Thus, in either case, the statement holds.
\end{proof}

Let $M_i$ be the matching that \cref{alg:randomsampling} returns. Note that in iteration $j$ we have $b = n^{2(j-1)\delta}$ and the algorithm takes $O(nb) = O(n^{1+2(j-1)\delta})$ time. So the running time of the algorithm until termination is $O(n) + O(n^{1+2\delta}) + \ldots + O(n^{1+2(i-1)\delta}) = O(n^{1+2(i-1)\delta})$. Since every vertex $v$ matched by $M_i$ must belong to $U'_{i-1}$, applying \cref{cl:huhu-hjb} gives $d^{i-1}_v = O(n^{1-2\delta(i-2)} \log n)$. This implies that the total number of edges in $G[V(M_i)]$ can be upper bounded by $|M| \cdot O(n^{1-2\delta(i-2)} \log n)$. Letting $d = \Theta(n^{1-2\delta(i-2)} \log n)$, the running time can be re-written as
$$
 O(n^{1+2(i-1)\delta}) = O(n^{1+2(i-1)\delta} d / d) = O(n^{1+2(i-1)\delta +1-2 \delta(i-2)} \log n / d) = O(n^{2+2\delta}\log n/d),
$$
which is the claimed bound.
\end{proof} 

Next, we apply \cref{prop:boosting-apx} on \cref{lem:nd} to complete the proof of \cref{lem:nd-boosted}.

\begin{proof}[Proof of \cref{lem:nd-boosted}] 

For any graph $H$, let \textsc{GreedyMatching}($H$) be a greedy algorithm that goes over all the edges of $H$ one by one and adds an edge $(u,v)$ to the matching if $u$ and $v$ are not matched already. Note that the output of this algorithm is a maximal matching in $H$.

The algorithm \textsc{MatchAndCertify} proceeds as follows. Our goal is to run the algorithm of \cref{prop:boosting-apx} to find a matching of size $\mu(G) - \epsilon n$ in $G$. This algorithm makes $t= O_\epsilon(1)$ adaptive calls 
$$
\mc{A}(G, U_1,\delta_1), \dots, \mc{A}(G, U_t,\delta_t)
$$
to an algorithm $\mathcal{A}$ that outputs a matching of size $\poly(\delta_i)\cdot \mu(G[U_i])$ in $G[U_i]$ provided that $\mu(G[U_i]) \geq \delta_i n$. Here $U_i$ is a subset of $V$ and parameter $\delta_i \leq \epsilon$ is a function of $\epsilon$. In order to use \cref{prop:boosting-apx}, we have to specify how the algorithm $\mc{A}$ is implemented.

We implement $\mathcal{A}(G, U_i, \delta_i)$ as follows. Let $G_i := G[U_i] \cap \Esparse$. As the first step, we call $\textsc{GreedyMatching}(G_i)$ to obtain matching $M_i$ in $O(|\Esparse|)$ time. If $|M_i| \geq \delta^2_i\cdot n / 8$, then we have a large enough matching in $G[U_i]$, and $\mathcal{A}$ terminates by outputting $M_i$ and we set $d_i = \infty$. Otherwise, we run $\textsc{RandomSampling}(G[U_i] \cap \Edense, \delta_i)$ from \cref{lem:nd}. Note that \cref{lem:nd} requires adjacency matrix access to its input graph which we can provide, given the guarantee of the lemma that we have adjacency matrix access to $\Edense$ and that the set $U_i$ is given explicitly. Since $\textsc{GreedyMatching}$ returns a 1/2-approximate maximum matching, we get that
$$
    \mu(G[U_i] \cap \Edense) \geq \mu(G[U_i]) - \mu(G[U_i] \cap \Esparse) \geq \delta_i n - 2|M_i| = 3\delta_i n /4.
$$
Therefore, algorithm $\textsc{RandomSampling}$ outputs a matching $M_{C_i}$ of size at least $9\delta^2_i n/16$ in $G[U_i] \cap \Edense$. By \cref{lem:nd}, this runs in $\widetilde{O}(n^{2+2\delta_i} /d_i) = \widetilde{O}(n^{2+\epsilon} /d_i)$ time where
$d_i = |E \cap V^2(M_{C_i})|/|M_{C_i}|$.

From our discussion above, in the $i$-th call to $\mathcal{A}$, this algorithm outputs a matching of size at least $\min\{ \delta_i^2 n /8, 9\delta_i^2 n/16\} = \delta_i^2 n /8$. Therefore, the algorithm of \cref{prop:boosting-apx} outputs a matching $M$ of size $\mu(G) - \epsilon n$. 

It remains to analyze the running time by specifying the outputs $d$ and $M_C$ and showing that at least one of the conditions \ref{itm:c1} or \ref{itm:c2} hold. 

If we never call the $\textsc{RandomSampling}$ algorithm, we simply set $d = 0$ and $M_C = \emptyset$ as all calls to $\mc{A}$. The total running time spent by \cref{prop:boosting-apx} preparing the vertex subsets $U_i$ for each $i \in [t]$, and the running time of the calls to the \textsc{GreedyMatching} is
 $$
 \widetilde{O}_\epsilon(n) + \sum_{i \in [t]} O(|E(G_i)|) = \widetilde{O}_\epsilon(n) + t|\Esparse| = \widetilde{O}_\epsilon(n + |\Esparse|).$$
Thus, condition \ref{itm:c1} holds. 

Otherwise, we let $j = \argmin_{i \in [t]} d_i$, set $d := d_j$ and return the certificate $M_C = M_{C_j} \setminus \Esparse$. We now show that condition \ref{itm:c2} holds.  

Let us now bound the running time. This is the cost of preparing the vertex subsets of \cref{prop:boosting-apx}, the running time of the calls to the \textsc{GreedyMatching}, and the running time of the calls to \textsc{RandomSampling}. The total running time is therefore $$\widetilde{O}_\epsilon(n) + t|\Esparse| +  \widetilde{O}\left(\sum_{i=1}^t \frac{n^{2+\epsilon}}{d_i}\right) = \widetilde{O}_\epsilon\left( |\Esparse| + \frac{n^{2+\epsilon}}{d}\right).$$

The bound above is the one required by \ref{itm:c2}, except that we defined $d = d_j = |E \cap V^2(M_{C_j})|/|M_{C_j}|$ whereas the certificate $M_C$ is a subset of $M_{C_j}$. To fix this, in what follows we prove that $d$ is at least a third of the value of the right parameter $d'' := |E \cap V^2(M_{C})|/|M_{C}|$, which completes the proof.

The fact that we run $\textsc{RandomSampling}$ on $G[U_j]$ implies that $\textsc{GreedyMatching}(G_j)$ did not return a large enough matching. This implies that $\mu(G_j) \leq \delta^2_j n/4$. Therefore, once we remove the edges of $\Esparse$ from $M_{C_j}$, we reduce the size of this matching by at most $\delta^2_j n/4$. Combined with $|M_{C_j}| \geq 9\delta_i^2 n/16$, this implies that $|M_C| \geq |M_{C_j}| - \delta_j^2/4 \geq 7|M_{C_j}|/16$. Hence
$$
d'' = \frac{|E \cap V(M_C)|}{|M_C|} \leq \frac{|E \cap V(M_{C_j})|}{7|M_{C_j}|/16} < 3d,
$$
completing the proof as discussed above.
\end{proof}

\subsection{Our Dynamic Algorithm via \cref{lem:nd-boosted}}\label{sec:dynamic}

In this section, we formalize our dynamic algorithm for \cref{thm:fully-dynamic} based on the static \textsc{MatchAndCertify} algorithm of \cref{lem:nd-boosted}. Our algorithm is formalized as \cref{alg:additive}. 

We prove in \cref{sec:correctness} that our algorithm maintains an additive $(1, \epsilon n)$-approximation of the graph at anytime.

\begin{lemma}[Correctness of \cref{alg:additive}]\label{lem:correctness}
    At any point, the output matching $M_G$ of \cref{alg:additive} is a $(1, \epsilon n)$ approximate maximum matching of graph $G$.
\end{lemma}

Later in \cref{sec:runtime}, we analyze the running time of the algorithm and prove the following bound on it.

\begin{lemma}[Runtime of \cref{alg:additive}]\label{lem:runtime}
    \cref{alg:additive} runs in time
    $$
        \widetilde{O}_\epsilon\left( \frac{t}{n} + \frac{n^{2+\epsilon} \cdot \ORS(\Theta_\epsilon(n))}{t} \right),
    $$ where $t$ is the threshold of the algorithm.
\end{lemma}

The combination of the two lemmas above with \cref{prop:multiplicative} implies \cref{thm:fully-dynamic}.

\begin{proof}[Proof of \cref{thm:fully-dynamic}]
The running time of \cref{alg:additive} is analyzed in \cref{lem:runtime}, which gives a boundo f
$$
\widetilde{O}_\epsilon\left( \frac{t}{n} + \frac{n^{2+\epsilon} \cdot \ORS_n(\Theta_\epsilon(n))}{t} \right),
$$
where $t$ is a threshold parameter.
To achieve the amortized time, we balance the two terms in the running time expression by setting
$$
 t = \sqrt{n^{3+\epsilon} \cdot \ORS_n(\Theta_\epsilon(n))}.
$$
By substituting this value of $t$ we have
$$
\widetilde{O}_\epsilon\left( \frac{\sqrt{n^{3+\epsilon} \cdot \ORS_n(\Theta_\epsilon(n))}}{n} + \frac{n^{2+\epsilon} \cdot \ORS(\Theta_\epsilon(n))}{\sqrt{n^{3+\epsilon} \cdot \ORS_n(\Theta_\epsilon(n))}} \right) = \widetilde{O}_\epsilon\left( \sqrt{n^{1+\epsilon} \cdot \ORS(\Theta_\epsilon(n))} \right).
$$

 By \cref{lem:correctness}, \cref{alg:additive} maintains a matching of size $\mu(G) - \epsilon n$ at any point.
 By applying \cref{prop:multiplicative}, we can boost this additive approximation to a multiplicative $(1-\epsilon)$-approximation
 while keeping the update time
$
\widetilde{O}_\epsilon \left( \sqrt{n^{1 + \epsilon} \cdot \ORS(\Theta_{\epsilon}(n))} \right).\qedhere
$
\end{proof}

\begin{figure}
\begin{algenv}{A dynamic $(1, \epsilon n)$-approximate matching algorithm.}{alg:additive}

    \nonl \textbf{Parameters:} $ \epsilon \in (0, 1]$, $t$ the threshold of the algorithm.

    Let $G=(V, E)$ be our input dynamic $n$-vertex graph.
    
    Let $\Gadd$, $\Gdel$, $\Hcert$ be the empty graphs on the vertex set $V$.

    We store all graphs $G, \Gadd, \Gdel, \Hcert$ both in adjacency list and matrix formats.

    $\cupdates \gets 0$ \Comment{\small {\color{gray} Counts the number of edge updates to $G$.}}

    $M_G \gets \emptyset$ \Comment{\small {\color{gray} This is our output matching.}} 
    
     \For{every edge update $e=(u, v)$}{
        
        \If{$e$ was inserted}{
            Add $e$ to $G$, $\Gadd$.

            Delete $e$ from $\Gdel$ (if it exists).
        }\ElseIf{$e$ was deleted}{
            Delete $e$ from $G$, $\Gadd$, $M_G$ (if it exists).
            
            Add $e$ to $\Gdel$.
        }

        $\cupdates \gets \cupdates + 1$

        \If{$\cupdates = t$}{
            Remove all edges in $\Gadd$, $\Gdel$, $\Hcert$. \label{line:phase-end}\Comment{\small {\color{gray} This ensures $|\Gadd|, |\Gdel|, |\Hcert| = O(t)$.}}
            
            $\cupdates \gets 0$
        }

        \If{$\cupdates \mod \frac{\epsilon n}{2} = 0$}{
           Run $\textsc{MatchAndCertify}(G - \Gadd, \Gadd \cup \Hcert - \Gdel, \epsilon/2)$ of \cref{lem:nd-boosted} and let $M$, $M_C$ be its outputs. \label{line:callMAC} \Comment{\small {\color{gray} Note that any adjacency matrix query to $G \setminus \Gadd$ can be answered with one adjacency matrix query to each of $G$ and $\Gadd$. Additionally, we can provide adjacency list access to $\Gadd \cup \Hcert - \Gdel$ in $|\Gadd \cup \Hcert \cup \Gdel| = O(t)$ time.}}

            Add $M_C$ to $\Hcert$.

            $M_G \gets M$.
        }

    }    
\end{algenv}
\end{figure}

\subsection{Correctness of \cref{alg:additive}}\label{sec:correctness}

In this section, we prove \cref{lem:correctness}. 

\begin{proof}[Proof of \cref{lem:correctness}]
     The algorithm works as follows. $\cupdates$ is a counter for the number of updates. Let a phase be $t$ consecutive updates starting with $\cupdates = 0$. The algorithm maintains the fully dynamic graph $G$. All the insertions throughout a phase  are added in $\Gadd$ (if an inserted edge is deleted, we delete it from $\Gadd$ too). All the deletions throughout a phase are added to $\Gdel$ (if a deleted edge is later inserted, we delete it from $\Gdel$ too).

     First, we prove that \textsc{MatchAndCertify} outputs a $(1,\epsilon n)$-approximate matching of graph $G$. To do this, we prove that the union of $\Edense = G -\Gadd$ and $\Esparse = \Gadd \cup \Hcert - \Gdel$ is exactly $G$. To see this, note that
     $$
     \Edense \cup \Edense = (G -\Gadd) \cup (\Gadd \cup \Hcert - \Gdel) = G \cup (\Hcert - \Gdel) = G.
     $$
     The second inequality holds because $\Gadd \cap \Gdel = \emptyset$. The last equality holds because any edge  $e \in \Hcert$ either belongs to $G$ or is deleted at some point, which means $\Hcert - \Gdel \subseteq G$. By \cref{lem:nd-boosted}, this implies that \textsc{MatchAndCertify} outputs a $(1,\epsilon n/2)$-approximate matching $M$ for $G$. Once $M$ is computed, we do not change $M$ within the next $\epsilon n / 2$ updates unless by applying the deletions on it. Each such deletion reduces the size of $M$ by at most 1. Therefore, after $\epsilon n/2$ updates, we have $|M| \geq \mu(G) - \epsilon n$, thus $M$ remains a $(1,\epsilon n)$-approximate matching for $G$ throughout.
\end{proof}

\subsection{Runtime Analysis of \cref{thm:fully-dynamic}}\label{sec:runtime}

Our goal in this section is to prove \cref{lem:runtime}. Namely, we want to bound the running time of \cref{alg:additive} by relating it to density of ORS graphs. The following \cref{clm:log n ORS} is our main tool for this connection to ORS graphs.

\begin{lemma}\label{clm:log n ORS}
    Let $G=(V, E = M_1 \cup \ldots \cup M_t)$ be an $n$-vertex graph where $M_1, \ldots, M_t$ are edge-disjoint matchings of size at least $r$ each. For any $i \in [t]$, define $G_i = (V, E_1 \cup \ldots \cup E_i)$ and let $d_i$ be $|E[G_i] \cap V^2(M_i)|/|M_i|$.  Then for any $\alpha \in (0,1)$, it holds that $$\sum_{i=1}^t \frac{1}{d_i} = O \left( \frac{1}{ \alpha^2} \cdot \ORS_n(r(1-\alpha)) \cdot \log n \right).$$ 
\end{lemma}

Let us first prove \cref{lem:runtime} via \cref{clm:log n ORS}. The proof of \cref{clm:log n ORS} is presented afterwards.

\begin{proof}[Proof of \cref{lem:runtime}]

All the steps of \cref{alg:additive} can be implemented in $\widetilde{O}(1)$ time per update except for \cref{line:phase-end} where we spend $\widetilde{O}(t)$ time and \cref{line:callMAC} where we call \textsc{MatchAndCertify}. Since we call \cref{line:phase-end} every $t$ updates, the amortized cost of this line is still $\widetilde{O}(1)$. It remains to analyze the cost of \cref{line:callMAC}. Note that \cref{line:callMAC} is called every $\epsilon n/2$ updates, so its cost will have to be amortized over these $\epsilon n/2$ updates.

Let us first discuss the time needed to prepare the inputs to $\textsc{MatchAndCertify}$. This input consists of graphs $\Edense = G - \Gadd$ and $\Esparse = \Gadd \cup \Hcert - \Gdel$. For $\Edense$ we only need to provide adjacency matrix query access. Since each such query can be answered by making one adjacency matrix query to $G$ and one to $\Gadd$, we do not need to spend any time preparing $\Edense$. For $\Esparse$, we go over the edges of $\Gadd, \Hcert, \Gdel$ one by one, but since these graphs are reset after $t$ updates, they only contain at most $t$ edges and the time spent is $O(t)$ which amortized over $\epsilon n$ updates adds up to $O_\epsilon(t/n)$.

Now, let us compute the running time of the algorithm based on the total calls to the \textsc{MatchAndCertify}. Each call to this algorithm, by \cref{lem:nd-boosted} satisfies one of the two conditions \ref{itm:c1} and \ref{itm:c2}. Each call that satisfies \ref{itm:c1} runs in $\widetilde{O}_\epsilon(n + |\Esparse|)$ time. Each call that satisfies \ref{itm:c2} runs in $\widetilde{O}_\epsilon(|\Esparse|+n^{2+\epsilon}/d)$ time. Let us for now ignore the second term in the latter running time. The rest of the computation, amortized over $\epsilon n$ updates, sums up to $\widetilde{O}_\epsilon(|\Esparse|/n) = \widetilde{O}_\epsilon(t/n)$.

Let us define a {\em phase} of the algorithm to be $t$ consecutive updates while $\cupdates < t$. For analyzing the second term in the running time of the calls to \textsc{MatchAndCertify} that satisfy condition \ref{itm:c2} of \cref{lem:nd-boosted}, we analyze their total cost over the whole phase and amortize this over $t$. Note that this is very different from our analysis before that amortizes the cost over a shorter sequence of $\epsilon n$ updates.

Take the $i$-th call to \cref{lem:nd-boosted} within the phase that satisfies \ref{itm:c2} and let $M_{C_i}$ be its certificate, which is a matching in $\Edense$ and let $d_i := |\Edense \cap V^2(M_{C_i})|/|M_{C_i}|$ be the average degree of $M_{C_i}$ in $\Edense$ the moment that it is returned. Note that the total cost the part of the algorithm that we ignored above can be written as
\begin{equation}\label{eq:bhgdpc-8213}
    \sum_{i} \widetilde{O}_\epsilon(n^{2+\epsilon}/d_i) = \widetilde{O}_\epsilon(n^{2+\epsilon}) \cdot \sum_{i} 1/d_i.
\end{equation}
In what follows, we upper bound the sum $\sum_i 1/d_i$.

Note that the set $\Edense$ used in the definition of $d_i$ does not remain the same throughout a phase, and particularly differs for each $d_i$. But here is the crucial observation. Since $\Edense = G - \Gadd$, the set $\Edense$ is \textbf{decremental} throughout the phase. This is because every edge inserted is inserted to both $G$ and $\Gadd$, thus $G - \Gadd$ remains unchanged, while deletions from $G$ will be deleted from $\Edense$ as well.

Now define $H_{\geq i} := M_{C_i} \cup M_{C_{i+1}} \cup \ldots \cup M_{C_k}$. From our earlier discussion that the $M_{C_i}$'s are matchings in a decremental graph, all edges in $H_{\geq i}$ were present when certificate $M_{C_i}$ was outputted by \cref{lem:nd-boosted}. Hence, for the average degree $d'_{\geq i}$ of $M_{C_i}$ in $H_{\geq i}$ we have
\begin{equation}\label{eq:hg-h23h}
d'_{\geq i} := \frac{|E(H_{\geq i}) \cap V(M_{C_i})^2|}{|M_{C_i}|} \leq \frac{|\Edense \cap V(M_{C_i})^2|}{|M_{C_i}|} = d_i,
\end{equation}
where here $\Edense$ is the input to \cref{lem:nd-boosted} at the step where certificate $M_{C_i}$ was returned. Therefore, applying \cref{clm:log n ORS} on the sequence $M_{C_k}, M_{C_{k-1}}, \ldots, M_{C_1}$ (note the reverse order), we have
$$
    \sum_{i=1}^k \frac{1}{d_i} \stackrel{\cref{eq:hg-h23h}}{\leq} \sum_{i=1}^k \frac{1}{d'_{\geq i}} = \widetilde{O}(\ORS_n(\Theta_\epsilon(n))).
$$
It should be noted that our \cref{clm:log n ORS} requires the matchings in the sequence to be edge-disjoint. This is indeed satisfied by our algorithm, since once a certificate $M_{C_i}$ is returned, we add its edges to $\Hcert$ and these edges will never belong to $\Edense$ for the rest of the phase.

Plugged into \cref{eq:bhgdpc-8213}, this bounds the total (remaining) running time of a phase by $$n^{2+\epsilon} \cdot \widetilde{O}(\ORS_n(\Theta_\epsilon(n)).$$ Amortized over $t$, the number of updates of the phase, this sums up to 
$$
\frac{n^{2+\epsilon} \cdot \widetilde{O}(\ORS_n(\Theta_\epsilon(n))}{t}.
$$ 
Combined with the earlier bound, we arrive at the claimed amortized update time of
$$\widetilde{O}_\epsilon\left(\frac{t}{n} + \frac{n^{2+\epsilon} \cdot \ORS_n(\Theta_\epsilon(n))}{t}\right).\qedhere
$$

\end{proof}

We now turn to prove \cref{clm:log n ORS}. First we prove the following auxiliary claim that helps with the proof of \cref{clm:log n ORS}.

\begin{claim}\label{clm: ORS d}
 Let $G=(V, E = M_1 \cup \ldots \cup M_t)$ be an $n$-vertex graph where $M_1, \ldots, M_t$ are edge-disjoint matchings of size at least $r$ each. For any $i \in [t]$, define $G_i = (V, E_1 \cup \ldots \cup E_i)$ and let $d_i$ be $|E[G_i] \cap V^2(M_i)|/|M_i|$. Let $d$ be an upper bound on $d_i$ for all $i \in [t]$.  Then for any $\alpha \in (0,1)$, it holds that $$t = O\left(\frac{d}{\alpha^2}\cdot\ORS_n(r(1-\alpha))\right).$$     
\end{claim}

\begin{proof}
    Let $\epsilon = \alpha/3$ and $\delta = \epsilon^2$. Let $H = M_1 \cup ... \cup M_{t}$. 
    Consider the following process for constructing a graph $H_3$, that we eventually show is an appropriate ORS graph.
    
    \begin{itemize}[leftmargin=10pt]
         \item \textbf{Step 1.} Let us define $x_e$ for any edge $e = (u,v)$ appearing in matching $M_{i}$ as follows 
        $$x_e := \sum_{ j > i: u,v \in V^2(M_j)}  \frac{1}{d_j}.$$
        Define $B_1 = \{e \mid x_e > 1/\epsilon \}$.
   
        Let $H_1 \subseteq H$ include a matching $M_i \in H$ iff no more than $2\epsilon r$ edges of $M_i$ belong to $B_1$. For each matching $M_i \in H_1$, we then set $M_i \gets M_i - B_1$. Observe that $|M_i| \geq (1-2\epsilon)r$ for any $M_i \in H_1$.

        \item \textbf{Step 2.} 

        Let $H_2 \subseteq H_1$ include every matching $M_i \in H_1$ independently with probability $\frac{\delta}{2d_i}$.

        \item \textbf{Step 3.} For any $i$, call an edge $e \in M_i \in H_2$ {\em bad} if $$e \in \bigcup_{j>i: M_j \in H_2} V^2(M_j).$$ 
        Let $B_2$ be the set of bad edges. Let $H_3 \subseteq H_2$ include matching $M_i \in H_2$ iff there are at least $(1-\epsilon)|M_i|$ edges in $M_i$ that do not belong to $B_2$. We then set $M_i \gets M_i - B_2$ for any $M_i \in H_3$. Observe that $|M_i| \geq (1-\epsilon)(1-2\epsilon)r \geq (1-3\epsilon)r$ for any $M_i \in H_3$.

    \end{itemize}

    From the construction above, it can be verified that $H_3$ is an ORS graph on matchings of size at least $(1-3\epsilon)r$. What remains, is to lower bound the number of matchings that survive in $H_3$. This is what the rest of the proof focuses on.

        First, we prove that $|B_1| \leq \epsilon|E|$. Suppose for the sake of contradiction that $|B_1| > \epsilon|E|$; then
        $$\sum_ {e \in B_1} x_e > |B_1|/\epsilon > |E|.$$
        On the other hand, we have
        $$
            \sum_{e \in B_1} x_e \leq \sum_{e \in E} x_e = \sum_{i \in [t]} \frac{|E[G_i] \cap V^2(M_i)|}{d} = \sum_{i \in [t]} \frac{d_i|M_i|}{d} \leq \sum_{i \in [t]} |M_i|  = |E|,
        $$
        which contradicts the inequality above. Here, the first equality follows by double counting: say matching $M_i$ ``contributes'' a value of $1/d$ to edge $e$ if $e \in E[G_i] \cap V^2(M_i)$; then $x_e$ can be verified to be the total contribution of all matchings to $e$. The second equality follows from the definition of $d_i$. The inequality after that follows from $d_i \leq d$, which is by the assumption of the lemma.

        Now that we know $|B_1| \leq \epsilon |E|$, we get that 
        \begin{equation}\label{eq:h1large}
            |H_1| \geq t/2.
        \end{equation}
        Otherwise, more than $t/2$ matchings have more than $2 \epsilon r$ edges in $B_1$, meaning that $|B_1| > \frac{t}{2} \cdot 2\epsilon r = \epsilon r t = \epsilon |E|$, a contradiction.

        Let us now fix a matching $M_i \in H_1$ and fix an edge $e=(u, v) \in M_i$. We have
        $$
            \Pr[e \in B_2 \mid M_i \in H_2] \leq \sum_{j > i: u,v \in V^2(M_j )}  \frac{\delta}{2d_j} = \delta x_e/2 \leq \frac{\delta}{2\epsilon} = \epsilon/2,
        $$
        where the last inequality follows because $x_e \leq 1/\epsilon$ for every edge in $M_i \in H_1$. By linearity of expectation, we have
        $$
        \E[|M_i \cap B_2| \mid M_i \in H_2] = \sum_{e \in M_i} \Pr[e \in B_2 \mid M_i \in H_2] \leq \frac{\epsilon}{2} |M_i|.
        $$
        Applying Markov's inequality, we thus have
        $$
        \Pr\big[|M_i \cap B_2| \geq \epsilon |M_i| \,\,\,\big\vert\,\,\, M_i \in H_2 \big] \leq 1/2.
        $$
        Therefore,
        $$
            \Pr[M_i \not\in H_3 \mid M_i \in H_2] \leq 1/2.
        $$
        This implies that
        $$
            \E[|H_3|] \geq \E[|H_2|]/2 = \frac{1}{2} \sum_{M_i \in H_1} \frac{\delta}{2d_i} = \frac{\delta}{4} \sum_{M_i \in H_1} \frac{1}{d_i} \geq \frac{\delta}{4} \sum_{M_i \in H_1} \frac{1}{d} = \frac{\delta |H_1|}{4d} \stackrel{\cref{eq:h1large}}{\geq} \frac{\delta t}{8d} = \Omega(\epsilon^2 t / d).
        $$

        Note that for each matching $M_i \in H_3$, we have $|M_i| \geq (1-\epsilon)(1-2\epsilon)r \geq (1-3\epsilon)r$ by construction. Additionally, $M_i$ is an induced matching in $\bigcup_{j<i; M_j \in H_3} M_j$, again, by construction. Therefore, $H_3$ is an ORS graph with matchings of size  $(1-3\epsilon)r$.  This implies that
        $$
        \ORS_n((1-3\epsilon)r) \geq  E[|H_3|] \geq \Omega(\epsilon^2 t /d).$$
        Moving the terms and recalling that $\epsilon = \alpha/3$, we get that
        $$
        t = O\left(\frac{d}{\alpha^2}  \cdot \ORS_n((1-\alpha)r) \right). \qedhere
        $$
\end{proof}

Now, we are ready to prove \cref{clm:log n ORS}.

\begin{proof}[Proof of \cref{clm:log n ORS}]
Given $M_1, \dots M_t$, we partition the matchings to $\log n$ groups $P_1, \dots ,P_{\log n}$ where $M_j$ is in the $i$-th group if $2^{i-1} \leq d_j <2^i$. Define $d'_j := |E[P_i] \cap V^2(M_j)|/|M_j|$.  Note that by definition for any matching $M_j \in P_i$ we have $d'_j \leq d_j \leq 2^i$. Therefore, by \cref{clm: ORS d} we have 
\begin{align*}
\sum_{i \in [\log n]} \sum_{j: M_j \in P_i} 1/d_j & \leq \sum_{i \in [\log n]} |P_i| /2^{i-1}  \\ & \leq O\left(\frac{1}{\alpha^2}\ORS_n((1-\alpha)r)\right) \cdot \sum_{i \in [\log n]} 2^i /2^{i-1} \tag{Applying \cref{clm: ORS d} with $d = 2^i$, we get $|P_i| \leq O\left(\frac{1}{\alpha^2}\ORS_n((1-\alpha)r) 2^i \right)$ which is plugged here.} \\&  = O\left(\frac{1}{\alpha^2}\ORS_n((1-\alpha)r) \cdot \log n\right).
\end{align*}
The proof is thus complete.
\end{proof}


\section{Bounding Density of (Ordered) Ruzsa-Szemer\'edi Graphs}\label{sec:ub}

In this section we prove upper bounds on the density of ORS graphs. Our main result of this section is a proof of \cref{thm:ORS-ub} which we present in \cref{sec:linear-matchings}. We then show that a much stronger upper bound can be proved if matchings cover more than half of vertices in \cref{sec:large-matchings}.

Before presenting our proof, let us first briefly discuss how the triangle removal lemma is useful for upper bounding density of RS graphs and why it does not seem to help for upper bounding ORS. The triangle removal lemma states that so long as the number of triangles in a graph are $o(n^3)$, then we can make the graph triangle free by removing $o(n^2)$ of its edges. Suppose we have a bipartite RS graph with induced matchings $M_1, \ldots, M_k$. Add $k$ vertices  $v_1, \ldots, v_k$ to this graph and connect each $v_i$ to all the vertices matched by $M_i$. Now, crucially, because each $M_i$ is an induced matching, each $v_i$ is part of exactly $|M_i|$ triangles, one for each edge of $M_i$. As such, the total number of triangles can be upper bounded by $O(kn) = O(n^2)$ which is well within the regime that one can apply triangle removal lemma. However, because the matchings in ORS are not induced matchings of the whole graph, the number of triangles cannot simply be upper bounded by $O(n^2)$ after adding the auxiliary vertices. In fact, our upper bound completely deviates from this approach and does not rely on the triangle removal lemma.

\subsection{Linear Matchings}\label{sec:linear-matchings}

The following lemma, which is one of our main tools in proving \cref{thm:ORS-ub}, shows that so long as vertex degrees are larger than a threshold, we can reduce the number of vertices rather significantly without hurting the size of a relatively large number of induced matchings. Our proof of \cref{lem:balanced-degrees} builds on the techniques developed for upper bounding RS graphs with matchings of size very close to $n/4$ in \cite{fox2015graphs} that we extend to ORS graphs with balanced degrees.

\begin{lemma}\label{lem:balanced-degrees}
    Let $G$ be an $\ORS_n(r,t)$ graph with $r =cn$ with even $t$. Let $M_1, \ldots, M_t$ be the $t$ matchings in $G$ as defined in \cref{def:ORS}. Let us partition these matchings into two subsets 
    $$
    \mc{M} = \{M_1, \ldots, M_{t/2}\} \qquad \text{and} \qquad \mc{M}' = \{M_{t/2 + 1}, \ldots, M_t\}.$$
    Suppose that for every vertex $v$, it holds that $\deg_{\mc{M}'}(v) \geq \delta ct$ for some $\delta > 0$. Then there exists an $\ORS_{n'}(cn, t')$ graph $H$ on $n' < (1-\delta c^2/2)n$ vertices where $t' \geq \delta c^2 t / (8 \cdot 2^x)$ and $x = 4n/ct$.
\end{lemma}

Let us start with an observation that is extremely helpful in proving \cref{lem:balanced-degrees}.

\begin{observation}\label{obs-main}
    Take a matching $M' \in \mathcal{M'}$ and take any vertex $v$ matched by $M'$. Let $N_{\mc{M}}(v)$ be the set of neighbors of $v$ in $\mc{M}$. That is, $u \in N_{\mc{M}}(v)$ iff there is $M \in \mc{M}$ such that $(v, u) \in M$. Then no vertex in $N_{\mc{M}}(v)$ can be matched in $M'$.
\end{observation}
\begin{proof}
    Suppose there is $u \in N_{\mc{M}}(v)$ that is matched by $M'$. Let $(u, w) \in M'$ be the matching edge involving $u$. Also take the edge $(v, y) \in M'$ that involves $v$ (which exists by definition of $v$). Note that since the matchings $M_1, \ldots, M_t$ are edge-disjoint by \cref{def:ORS}, $v$ and $u$ cannot be matched together in $M'$ (as $v$ and $u$ must be matched in some other matching in $\mc{M}$ by definition of $N_\mc{M}(v)$). Now since $(u, w), (v, y) \in M'$ but $(v, u)$ belongs to some matching in $\mc{M}$ that comes before $M'$ in the ordering, matching $M'$ cannot be an induced matching among its previous matchings,  contradicting \cref{def:ORS} for ORS graphs. 
\end{proof}

We are now ready to prove \cref{lem:balanced-degrees}.

\begin{proof}[Proof of \cref{lem:balanced-degrees}]
    We explain the proof in a few steps.
    
    \paragraph{Step 1: The pivots.} We iteratively take a vertex of highest remaining degree in $\mathcal{M}$, and remove its neighbors in $\mathcal{M}$ from the graph.  More formally, take the graph $G_0 = (V_0, E_0)$ where $V_0 = V(G)$ and $E_0$ is the union of the matchings in $\mc{M}$. At step $i$, we choose $v_i$ in $V_{i-1}$ with the maximum degree in $G_{i-1}$. Let us define $N_i$ as the neighbors of $v_i$ in $G_{i-1}$. We then remove $v_i$ and its neighbors in $G_{i-1}$ to define $G_i=(V_i, E_i)$. Namely, $V_i = V_{i-1} \setminus (\{v_i \} \cup N_i)$ and $E_i$ includes all edges in $E_{i-1}$ except those that have at least one endpoint that does not belong to $V_i$.
    
    Let $k$ be the maximum integer such that $|N_k| > n/ x$. We define $P = \{v_1, \ldots, v_k\}$ to be the set of {\em pivots}. Note that since $N_1, \dots, N_k$ are disjoint sets, we get that $ n \geq \sum_ {i=1} ^k |N_i| \geq kn/x$ which implies that 
    \begin{equation}
        |P| \leq x.
    \end{equation}
    
    Moreover, note that when removing $v_i$, we remove $|N_i|$ vertices from the graph and each of those vertices has remaining degree at most $|N_i|$ (as we choose $v_i$ to be the vertex of highest remaining degree). In other words, we remove at most $|N_i|^2$ from the graph after removing $v_i$ and its remaining neighbors. On the other hand, after removing $v_1, \ldots, v_k$ and their neighbors, the maximum degree in the graph is at most $n/x$ (by definition of $k$), and so there are at most $n^2/x$ edges in the graph. This implies that:
    $$
        |E_0| - \sum_{i=1}^k |N_i|^2 \leq n^2/x.
    $$
    Given that $E_0$ includes $t/2$ edge disjoint matchings of size $cn$, we get that $|E_0| \geq tcn/2$. Plugged into the inequality above, this implies that
    \begin{equation}\label{eq:Nisquared}
        \sum_{i=1}^k |N_i|^2 \geq |E_0| - n^2/x \geq tcn/2 - n^2/x.
    \end{equation}

    \paragraph{Step 2: Vertex reduction.}  Take a matching $M \in \mathcal{M}'$. Let $S_M$ be the set of pivots matched by $M$, i.e., $S_M = P \cap V(M)$. Let us define $Y_M := V \setminus \cup_{v_i \in S} N_i$. From \cref{obs-main}, we get that all edges of $M$ must have both endpoints still in $Y_M$. Suppose for now that $M \in \mathcal{M}'$ is chosen uniformly at random. We first argue that the expected size of $|Y_M|$ is relatively small (despite it preserving all edges of $M$ completely). We have
    \begin{flalign}
        \nonumber \E_{M \sim \mc{M}'}[|Y_M|] &= n - \sum_{i=1}^k \Pr[v_i \in V(M)] \cdot |N_i|\\
        \nonumber &= n - \sum_{i=1}^k \frac{\deg_{\mc{M}'}(v_i)}{|\mc{M}'|} \cdot |N_i| \tag{As $M$ is chosen uniformly from $\mc{M}'$ and $v_i$ is matched in $\deg_{\mc{M}'}(v_i)$ of matchings in $\mc{M}'$.}\\
        \nonumber &= n - \sum_{i=1}^k \frac{2\deg_{\mc{M}'}(v_i)}{t} \cdot |N_i| \tag{Since $|\mc{M}'| = t/2$.}\\
        \nonumber &\leq n - \sum_{i=1}^k \frac{2\delta c t}{t} \cdot |N_i| \tag{Since $\deg_{\mc{M}'}(v_i) \geq \delta c t$ as assumed in \cref{lem:balanced-degrees}.}\\
        \nonumber &\leq n - \sum_{i=1}^k \frac{4\delta c}{t} \cdot |N_i|^2. \tag{Since $|N_i| \leq t/2$ as $\mc{M}$ is the union of $t/2$ matchings.}\\
        \nonumber &\leq n - \frac{4\delta c}{t} \cdot (tcn/2 - n^2/x) \tag{By \cref{eq:Nisquared}.}\\
        \nonumber &= n - 2 \delta c^2 n + \frac{4\delta c n^2}{tx}\\
        \nonumber &\leq n - 2\delta c^2 n + \frac{4\delta c n^2}{(4n/cx)x} \tag{Since $t = 4n/cx$ by definition of $x$ in \cref{lem:balanced-degrees}.}\\
        \nonumber &= n - 2\delta c^2 n + \delta c^2 n\\
        &= (1 - \delta c^2)n.\label{eq:exp-YM-ub}
    \end{flalign}

    Instead of the expected value of $Y_M$, we need some (weak) concentration. Take $\epsilon = \delta c^2/2$ and note from \cref{eq:exp-YM-ub} that $(1+\epsilon)\E[|Y_M|] < (1+\epsilon)(1-\delta c^2)n < (1-\delta c^2/2) n$. We have
    \begin{flalign}
        \nonumber \Pr_{M \sim \mc{M}'}\Big[|Y_M| < (1+\epsilon)\E[|Y_M|] < (1-\delta c^2/2)n \Big] &= 1 - \Pr_{M \sim \mc{M}'}\Big[|Y_M| > (1+\epsilon)\E[|Y_M|]\Big] \\
        \nonumber &\geq 1 - \frac{1}{1+\epsilon}\tag{By Markov's inequality.}\\
        &\geq \epsilon/2 = \delta c^2/4.\label{eq:markov-YM}
    \end{flalign}
    
    We call a matching $M \in \mc{M}'$ a good matching if $|Y_M| < (1-\delta c^2 / 2)n$. From \cref{eq:markov-YM} we get that at least $(\delta c^2 / 4) |\mc{M}'| = \delta c^2 t / 8$ matchings are good.

    \paragraph{Step 3: Grouping good matchings.} Now, we partition the good matchings in $\mathcal{M}'$ into $2^{|P|} \leq 2^x$ classes depending on which subset of the pivots they match. Since there are at least $\delta c^2 t / 8$ good matchings, at least $\frac{\delta c^2 t}{8 \cdot 2^x}$ of them must belong to the same class $C$. Let $Y = Y_M$ for some $M \in C$. Note also that $Y = Y_{M'}$ for any other matching $M' \in C$ as well since they all match the same subset of pivots. Moreover, we have $|Y| < (1-\delta c^2/2)n$ since $M$ is good, and all matchings in $C$ only match vertices in $Y$ as discussed earlier (as a consequence of \cref{obs-main}). This implies that the graph $H$ on vertex set $Y$ and edge-set $C$ (i.e., including all edges of all matchings in $C$) is an $\ORS_{n'}(cn, t')$ graph for $n' < (1-\delta c^2/2)n$ and $t' \geq \delta c^2 t / (8 \cdot 2^x)$ as desired.
\end{proof}

The next lemma follows by applying \cref{lem:balanced-degrees} after carefully pruning vertices of low degree in $\mc{M}'$. Intuitively, \cref{lem:two-cases} significantly reduces the number of vertices and shows that a relatively large number of matchings will have a lot of edges in the remaining graph.

\begin{lemma}\label{lem:two-cases}
    Let $G$ be an $\ORS_n(cn, t)$ graph for even $t$. Then for some $\kappa \geq c^3/160$, there exist $\ORS_{n'}(c'n', t')$ graphs such that $n' \leq n$, 
      $c' \geq (1+\kappa)c$ and $t' \geq t/2^{O(n/ct)}$.
\end{lemma}
\begin{proof}
    Let $\mc{M}$ and $\mc{M}'$ be defined for $G$ as in \cref{lem:balanced-degrees}.  We would like to apply \cref{lem:balanced-degrees} iteratively to prove \cref{lem:two-cases}. However, \cref{lem:balanced-degrees} requires a lower bound on the degrees in $\mc{M}'$ which might not necessarily hold for our graph $G$. Therefore to be able to use it, we first have to prune vertices of small degrees.
    
    Let $G'_0$ be the subgraph of $G$ only including the edges of the matchings in $\mc{M}'$. Let $\delta = 1/4$. We iteratively remove vertices of low degree. Formally, in step $i$, if there is a vertex $v_i$ with degree less than $2\delta c t$ in $G'_{i-1}$, we define the graph  $G'_i$ of the next step to be graph obtained after removing $v_i$ and its edges from $G'_{i-1}$. Let $G'_k$ be the final graph that does not have any vertex of degree smaller than $2\delta c t$. 

Suppose that $k = \alpha n$; that is, we remove $\alpha n$ vertices in the process above. We consider two cases depending on the value of $\alpha$.

\paragraph{Case 1: $\alpha \geq \delta c^3/10$.} Since we remove at most $\delta c t$ edges in every step, the total number of removed edges is at most $\delta \alpha  c t n$ over the $k=\alpha n$ steps of constructing $G'_k$. Say a matching $M \in \mc{M}'$ is {\em damaged} if a total of at least $3 \delta \alpha cn$ of its edges have been removed. In other words, $M$ is damaged if less than $|M| - 3\delta \alpha c n = (1-3\delta \alpha)cn$ of its edges belong to $G'_k$. Since the total number of edges removed is at most $\delta \alpha c t n$ and each damaged matching has at least $3\delta \alpha c n$ of its edges removed, the total number of damaged matchings can be upper bounded by $\delta \alpha c t n / 3\delta \alpha c n = t/3$. Since $|\mc{M}'| = t/2$, at least $t/6$ matchings in $\mc{M}'$ are not damaged, and have size at least $(1-3\delta \alpha)c n$. These matchings themselves form an $\ORS_{n'}(c' n', t/6)$ graph with parameters
$$
n' = (1-\alpha)n \leq (1-\delta c^3/10)n \stackrel{(\delta=1/4)}{=} (1- c^3/40)n,
$$
and
$$
c' n' \geq (1-3\delta \alpha)cn = (1-3\delta \alpha) c \frac{n'}{1-\alpha} \stackrel{(\delta=1/4)}{=} \left(1-\frac{3}{4}\alpha \right) c \frac{n'}{1-\alpha} \geq \left(1 + \frac{1}{4}\alpha \right) c n' \geq (1+c^3/160)cn',
$$
which implies that $c' \geq (1+c^3/160)c$. Taking $\kappa = c^3/160$ proves the lemma in this case. 


\paragraph{Case 2 -- $\alpha < \delta c^3/10$:} Let $G_k$ be the graph $G$ after removing vertices $v_1, \ldots, v_k$ (the difference between $G_k$ and $G'_k$ is that $G'_k$ only includes the edges of $\mc{M}'$ but $G_k$ includes both edges of $\mc{M}'$ and $\mc{M}$ that do not have any endpoint removed). Since we remove a total of $\alpha n$ vertices from $G$ to obtain $G_k$, each matching in $\mc{M}'$ and $\mc{M}$ loses at most $\alpha n$ edges. Therefore, $G_k$ is an $\ORS_{n''}(c''n'', t'')$ graph with parameters 
\begin{equation}\label{eq:hc82s37d}
n'' = (1-\alpha)n, \qquad
c'' \geq \frac{cn - \alpha n}{n''} = \frac{(c-\alpha)n}{(1-\alpha)n} \geq c - \alpha \geq (1-\delta c^2/10)c, \qquad
t''=t.
\end{equation}
Additionally, by the construction of $G'_k$, all vertices $v$ in $G_k$ satisfy 
\begin{flalign*}
    \deg_{\mc{M}'}(v) \geq 2 \delta ct \geq \delta c'' t'' \tag{Since $t'' = t$ and $c'' \leq \frac{cn}{n''} = \frac{cn}{(1-\alpha)n} \leq 2c$ where the last inequality follows from $\alpha < 1/2$.}
\end{flalign*}
This now satisfies the requirements of \cref{lem:balanced-degrees}. Applying it on graph $G_k$, we obtain an $\ORS_{n'}(c'n', t')$ graph where the number of vertices satisfies 
\begin{flalign}\label{eq:n-haoeunh}
    n' &< (1-\delta c''^2/2)n'',
\end{flalign}
the number of matchings satisfies
$$
 t' \geq \delta c''^2 t'' / (8 \cdot 2^{4n''/c''t''}) = \delta c \cdot t / 2^{O(n/c t)} = t/2^{O(n/ct)},
$$
and finally since \cref{lem:balanced-degrees} does not change the size of matchings, we get
\begin{flalign*}
 c' n' &= c'' n''\\
 &\geq \frac{c''n'}{(1-\delta c''^2/2)} \tag{By \cref{eq:n-haoeunh}.}\\
 &\geq (1+\delta c''^2/2)c'' n'\\
 &\geq \Big(1+\delta \Big((1-\delta c^2/10)c\Big)^2/2\Big)\Big((1-\delta c^2/10)c \Big) n' \tag{Since the RHS is minimized when $c''$ is minimized, thus we can replace $c''$ with its LB from \cref{eq:hc82s37d}.},\\
 &\geq \Big(1+0.4 \delta c^2\Big)\Big((1-\delta c^2/10)c \Big) n' \tag{Since $\delta\big((1-\delta c^2/10)c\big)^2/2 \geq \delta (0.9c)^2/2 > 0.4\delta c^2$}\\
 &> (1+0.07 c^2)cn'.
\end{flalign*}
Dividing both sides of the inequality by $n'$ implies $c' \geq (1+0.07c^2)c$. Taking $\kappa = 0.07c^2 \geq c^3/160$ implies the lemma.
\end{proof}

We are now ready to prove \cref{thm:ORS-ub}.

\begin{proof}[Proof of \cref{thm:ORS-ub}]
We prove that for some $d = \poly(1/c)$, there does not exist any  $\ORS_n(cn, t)$ graph $G$ with $t \geq n/\log_b^{(d)}n$ where $b = 2^{\beta/c}$ for some large enough constant $\beta \geq 1$. Note that the base of the iterated log is different from the statement of the theorem, but since $\log^{(x)}_b z = \log^{(\Theta(x \log^* b))} z$, this implies the theorem as well by letting $\ell = \Theta(d \cdot \log^* b) = \poly(1/c)$.

Suppose for the sake of contradiction that that there exists an $\ORS_n(cn, t)$ graph $G$ with $t \geq n/\log_b^{(d)}n$. Let $\kappa$ be as in \cref{lem:two-cases}. We iteratively apply \cref{lem:two-cases} for $k = \log_2(1/c)/\kappa = \poly(c)$ steps to obtain a sequence of graphs $G_0, G_1, G_2, G_3, \ldots, G_k$ where $G_0 = G$ is the original graph, $G_1$ is obtained by applying \cref{lem:two-cases} on $G_0$, $G_2$ is obtained by applying \cref{lem:two-cases} on $G_1$, and so on so forth.

Let us now analyze the ORS properties of graph $G_k$, starting with the parameter $c_k$. Since every application of \cref{lem:two-cases} multiplies the parameter $c_i$ by a factor of at least $(1+\kappa)$, we have
\begin{flalign*}
c_k &\geq (1+\kappa)^{k} c\\
&= (1+\kappa)^{\log_2(1/c)/k} c\\
&\geq 2^{\log_2(1/c)}c \tag{Since $(1+\kappa)^{1/\kappa} \geq 2$ for all $0 < \kappa \leq 1$}\\
&\geq 1.
\end{flalign*}
This implies that in graph $G_k$, there must be $t_k$ edge-disjoint matchings of size $c_k n_k \geq n_k$, but each matching in a graph on $n_k$ vertices can have size at most $n_k/2$. In other  words, we must have $t_k = 0$. We will obtain a contradiction by proving that $t_k \geq 1$. 

We prove by induction that for every $i \in [d]$, 
$$
t_i \geq \frac{n}{\log^{(d-i)}_b n}.
$$ 

For the base case $i=0$ this holds as assumed at the beginning of the proof. By choosing appropriately large $\beta$, we know by \cref{lem:two-cases} that,
\begin{flalign*}
    t_i &\geq t_{i-1}/2^{0.5\beta(n_{i-1}/c_{i-1} t_{i-1})}\\
    &\geq t_{i-1}/2^{0.5\beta(n/c t_{i-1})} \tag{Since $n_{i-1} \leq n$ and $c_{i-1} \geq c$ as we iteratively apply \cref{lem:two-cases}.}\\
    &= t_{i-1}/b^{0.5(n/ t_{i-1})} \tag{Since we defined $b = 2^{\beta/c}$.}\\
    &\geq \frac{n}{\log^{(d-i+1)}_b n} \cdot b^{-0.5\left(n/\frac{n}{\log^{(d-i+1)}_b n}\right)} \tag{By the induction hypothesis $t_{i-1} \geq \frac{n}{\log^{(d-i+1)}_b n}$.}\\
    &= \frac{n}{\log^{(d-i+1)}_b n} \cdot b^{-0.5\log^{(d-i+1)}_b n}\\
    &= \frac{n}{\log^{(d-i+1)}_b n \cdot \sqrt{ \log^{(d-i)}_b n}}\\
    &\geq \frac{n}{\log^{(d-i)}_b n},
\end{flalign*}
concluding the proof.

Now let $d > k = \poly(1/c)$. From the above, we get that $t_k \geq n/(\log_b^{(d-k)}) \geq 1$, which as discussed is a contradiction.
\end{proof}

\subsection{When Matchings are (Very) Large}\label{sec:large-matchings}

In this section, we prove that the number of matching of an ORS graph with $r > n/4$ is bounded by a constant. More precisely we prove \cref{thrm-const}. The same upper bound was already known for RS graphs \cite{fox2015graphs}, and we show that nearly the same proof carries over to ORS graphs as well.

\begin{theorem} \label{thrm-const}
    For any $c > 1/4+\epsilon$, $\ORS_n(cn) \leq 1/\epsilon + 1$.
\end{theorem}

\begin{proof}
Suppose $G$ is an $\ORS(r,t)$ graph on $n$ vertices and let $M_1, \ldots, M_t$ be the corresponding ordered list of matchings as in \cref{def:ORS}. We show that for any $i \not= j$ in $[t]$, we have $|V(M_i) \cap V(M_j)| \leq r$. The rest of the proof follows exactly as in the upper bound of \cite[Section~2]{fox2015graphs} for RS, which we omit here.

Let us assume w.l.o.g. that $i<j$ and suppose that $|V(M_i) \cap V(M_j)| \geq r+1$. By the pigeonhole principle, this means that there must be an edge in $M_i$ whose both endpoints are matched in $M_j$, contradicting the fact from \cref{def:ORS} that $M_j$ is an induced matching in $M_1 \cup \ldots \cup M_j$. Hence, $|V(M_i) \cap V(M_j)| \leq r$ and the proof is complete.
\end{proof}

\section*{Acknowledgements}

We thank 
Sepehr Assadi, 
Sanjeev Khanna,
Huan Li,
Ray Li,
Mohammad Roghani, and
Aviad Rubinstein for helpful discussions about RS graphs and their applications in dynamic graphs over the years. Additionally, we thank Sepehr Assadi, Jiale Chen, and anonymous FOCS'24 reviewers for helpful suggestions on improving the exposition of the paper.


\bibliographystyle{plainnat}

\bibliography{references}
	
\end{document}